%% file: main.tex
\documentclass[letterpaper, 10 pt, conference]{ieeeconf} 
\input{preamble.tex}

\newcommand{\real}{\mathbb{R}} 

\newcommand{\flow}{\bvarphi}

\title{\textbf{Safety Filtering with an Infinite Number of Constraints}}
\author{Max H. Cohen$^1$, Pio Ong$^2$, Pol Mestres$^2$, and Aaron D. Ames$^2$%
\thanks{$^1$ The author is with the Department of Electrical and Computer Engineering, North Carolina State University, Raleigh, NC \texttt{mhcohen2@ncsu.edu}.}
\thanks{$^2$ The authors are with the Department of Mechanical and Civil Engineering, California Institute of Technology, Pasadena, CA
\texttt{pioong,mestres,ames@caltech.edu}.
}
}

\begin{document}

\maketitle
\begin{abstract}
    Control barrier functions (CBFs) provide a rigorous framework for designing controllers enforcing safety constraints. While CBF theory is well-developed for a finite number of safety constraints, certain applications, e.g., backup CBFs, require an infinite number of constraints. Despite the practical success of CBFs, several fundamental questions remain unanswered when safe sets are defined with an infinite numbers of constraints, including: necessary and sufficient conditions for forward set invariance, the actual definition of CBFs associated with these sets, the regularity properties of the resulting controllers, and the ability to reduce a collection of infinite constraints to a finite number. This paper addresses these questions by extending CBF theory to the infinite constraint setting. We identify regularity conditions under which Nagumo's Theorem reduces to barrier-like inequalities and when the associated CBF controllers are at least continuous. We further connect these results to optimal-decay CBFs, bridging theoretical conditions for invariance and practical instantiations of the resulting controller. Finally, we illustrate how the developed theory addresses limitations of backup CBFs.
\end{abstract}

\section{Introduction}
Control barrier functions (CBFs) \cite{AmesTAC17} have been widely used to design controllers enforcing safety constraints on modern autonomous systems. 
The theoretical foundations of CBFs are rooted in the notion of set invariance \cite{Blanchini}, and, in particular, Nagumo's Theorem \cite{nagumo1942lage}, which provides necessary and sufficient conditions for set invariance. In the broader literature on safety-critical control, encompassing not just CBFs but also techniques such as model predictive control \cite{WabersichTAC22} and reachability analysis \cite{HerbertCDC21}, Nagumo's Theorem is frequently invoked to relate properties of the dynamics on the boundary of a prescribed safe set to invariance of that set. Yet many of these stated conditions are only equivalent to Nagumo's Theorem when the dynamics and the set itself satisfy certain regularity conditions. 

CBFs now provide a comprehensive theory of set invariance for safe sets defined as the zero superlevel set of a single continuously differentiable function $h(\bx)\geq0$ and dynamics described by a locally Lipschitz vector field $\dot{\bx}=\bF(\bx)$. Under the assumption that zero is a regular value of $h$, Nagumo's Theorem is equivalent to the often stated condition: $h(\bx)=0$ implies $\dot{h}=\pdv{h}{\bx}\bF(\bx)\geq0$; see \cite[Ex. 4.1.29]{AbrahamMarsdenRatiu}. When zero is not a regular value of $h$, this condition is no longer equivalent to Nagumo's Theorem nor is it meaningful for set invariance unless additional assumptions are placed on the behavior of $\dot{h}$ outside of the safe set via extended class $\mathcal{K}$ functions \cite{KondaLCSS21}. 

Given that many real world systems are subject to multiple safety constraints, there is also a healthy body of work on CBF theory for sets defined by multiple constraints $h_i(\bx)\geq0$ for $i\in\{1,\dots,N\}$. Early approaches focused on composing multiple CBFs into a single (candidate) CBFs via Boolean combinations \cite{GlotfelterLCSS17}. While such approaches reduce a large collection of CBF constraints to a single constraint, the resulting CBF is generally non-smooth, leading to discontinuous controllers. An alternative approach is to directly consider multiple CBFs \cite{XuAutomatica18,BreedenACC23,IsalyTAC24,CohenCDC25}, leading to controllers subject to multiple CBF constraints.
For certain classes of dynamics and constraints, these multi-CBF controllers can be shown to be locally Lipschitz \cite{XuAutomatica18,CohenCDC25}. More generally, continuity of multi-CBF controllers can be guaranteed by imposing various constraint qualification conditions \cite{IsalyTAC24,PO-BC-LS-JC:23-auto,PM-AA-JC:25-ejc}.

Despite the maturity of the multiple CBF literature, there are applications that theoretically require \emph{infinitely} many constraints. The most prominent example is the backup CBF framework \cite{gurriet2020scalable,GurrietICCPS18,YuxiaoCDC21,TamasACC23}, where a conservative safe set---equipped with a backup controller that renders it invariant---is expanded to a larger safe set by leveraging the flow of the system under the backup controller. In this setting, the resulting invariant set and corresponding CBF conditions are characterized by infinitely many constraints---one for each point in time along the horizon of the backup flow. Other applications include characterizing collision-free configurations of rigid bodies~\cite{Thirugnanam22}, or enforcing safety under state uncertainty~\cite{DeanCoRL20} wherein the safe set is defined by the configurations in which all state estimates are safe.

While CBF techniques with infinite constraints have enjoyed practical success, fundamental theoretical questions remain unanswered: What are necessary and sufficient conditions for forward invariance of these sets? When do the resulting controllers enjoy desirable regularity properties, such as continuity? Can the infinite collection of constraints be reduced to a finite collection whose satisfaction implies that of the infinite collection? In this paper, we answer these questions by extending fundamental results on CBF theory to the infinite constraint setting. We first identify regularity conditions that reduce Nagumo's Theorem for set invariance to  barrier-like inequalities. We then connect these results to CBFs and \emph{optimal-decay} CBFs (OD-CBFs) \cite{ZengACC21,PioCDC25}. %which obviates the need to search for an appropriate class $\mathcal{K}$ function in the corresponding CBF conditions. 
We also outline results guaranteeing continuity of optimization-based controllers and discuss when invariance is preserved during reduction to finitely many constraints. We finally connect our results to backup CBFs, illustrating how our framework overcomes certain limitations of backup CBFs.

\section{Background and Problem Formulation}\label{sec:prelim}

In this paper we consider nonlinear control affine systems:
\begin{equation}\label{eq:control-affine}
    \dot{\bx} = \bf(\bx) + \bg(\bx)\bu,
\end{equation}
with state $\bx\in\R^n$ and input $\bu\in\R^m$ (the results in the paper extend for arbitrary convex input sets $\Uc$, but we keep $\Uc = \R^m$ for simplicity). The drift $\bf\,:\,\R^n\rightarrow\R^n$ and control directions $\bg\,:\,\R^n\rightarrow\R^{n\times m}$ are assumed locally Lipschitz. Our main objective is to design a controller $\bk\,:\,\mathcal{D}\rightarrow\real^m$ for \eqref{eq:control-affine} that renders a desired \emph{safe set} $\mathcal{S}\subset\mathcal{D}\subseteq\R^n$ forward invariant for the closed-loop system:
\begin{equation}\label{eq:closed-loop}
    \dot{\bx} = \bf(\bx) + \bg(\bx)\bk(\bx) \eqqcolon \bF(\bx).
\end{equation}
A common approach to designing controllers enforcing the forward invariance of safe sets is CBFs. Here, one considers candidate safe sets of the form:
\begin{equation}\label{eq:S}
    \begin{aligned}
        \mathcal{S} = & \{\bx\in\R^n\mid h(\bx) \geq 0\}, \\
        \partial\mathcal{S} = & \{\bx\in\R^n\mid h(\bx) = 0\},
    \end{aligned}
\end{equation}
where $h\,:\,\R^n\rightarrow\R$ is continuously differentiable. Whether \eqref{eq:S} can be rendered invariant depends on whether $h$ is a CBF.

\begin{definition}[\cite{AmesTAC17}]\label{def:cbf}
    A continuously differentiable function $h\,:\,\R^n\rightarrow\R$ defining a set $\mathcal{S}\subset\R^n$ as in \eqref{eq:S} is said to be a control barrier function for \eqref{eq:control-affine} on the set $\mathcal{S}$ if there exists a $\mathcal{K}$ function $\alpha\,:\,\R_{\geq 0}\rightarrow\R$ such that:
    \begin{equation}\label{eq:CBF}
        \sup_{\bu\in\real^m}\left\{\pdv{h}{\bx}(\bx)\bf(\bx) + \pdv{h}{\bx}(\bx)\bg(\bx)\bu \right\} \geq - \alpha(h(\bx)),~ \forall\bx\in\mathcal{S}.
    \end{equation}
     Given an open set $\Dc \supset \Sc$, $h$ is a CBF on $\Dc$ if there exists an extended class $\Kc$ function $\alpha$ such that~\eqref{eq:CBF} holds for $\bx\in\Dc$.
\end{definition}

CBFs help certify the existence of $\bx\to\bk(\bx)$ satisfying:
\begin{equation}\label{eq:CBF-condition}
    \pdv{h}{\bx}(\bx)\bf(\bx) + \pdv{h}{\bx}(\bx)\bg(\bx)\bk(\bx) \geq -\alpha (h(\bx)),
\end{equation}
which can be imposed either on $\Sc$ or a larger domain $\Dc\supset\Sc$, each offering distinct advantages. When imposed only on $\Sc$, 
% the resulting condition 
\eqref{eq:CBF-condition}
admits a direct relationship with forward invariance. When $\Sc$ is \emph{well-posed}, 
% which in this setting corresponds to 
in that
$0$ is a regular value of $h$, forward invariance of $\Sc$ is equivalent to~\eqref{eq:CBF-condition} holding on~$\partial \Sc$, under the assumption that the closed-loop vector field $\bF$ is locally Lipschitz. When \eqref{eq:CBF-condition} is imposed on a larger domain $\Dc \supset \Sc$, the Lipschitz assumption can be dropped, and asymptotic stability of $\Sc$ is also guaranteed
%~\cite[Thm. 2]{GlotfelterLCSS17} 
\cite{KondaLCSS21}.

Although Definition~\ref{def:cbf} assumes that $\Sc$ is defined as the $0$-superlevel set of a single continuously differentiable function, in practice the safe set is often defined through a finite collection of functions $\{h_i\}_{i=1}^N$, where $N\in\mathbb N$, i.e.:
\begin{align*}
    \Sc = \setdef{\bx\in\real^n}{h_i(\bx) \geq 0, \ i\in[N]}.
\end{align*}
In this case, one seeks to design a controller that satisfies the CBF conditions associated with all of the $h_i$.

However, there are cases where the safe set is defined by an \emph{infinite} number of continuously differentiable functions. This arises, for instance, in \emph{backup CBFs}~\cite{gurriet2020scalable,GurrietICCPS18,YuxiaoCDC21,TamasACC23}, where a conservative safe set
% (equipped with a backup controller that renders it safe)
is expanded to a larger safe set by leveraging the flow of the system under the backup controller,
in addition to other cases previously mentioned \cite{Thirugnanam22,DeanCoRL20}.

In this paper, we tackle this scenario by considering a compact set $\Tc\subset\real^p$, a continuously differentiable function $h:\Tc\times\real^n\to\real$, and a safe set defined as: 
\begin{align}\label{eq:safe-set-infinite-constraints}
    \Sc = \setdef{\bx\in\real^n}{h(\tau,\bx) \geq 0, \forall \tau \in \Tc}.
\end{align}

Although various papers,
% in the literature, 
particularly those on backup CBFs, routinely search for controllers for~\eqref{eq:safe-set-infinite-constraints} by imposing: 
% the inequalities:
\begin{align}\label{eq:cbf-conditions-infinite}
    \frac{\partial h}{\partial \bx}(\tau,\bx) (\bf(\bx)+\bg(\bx)\bu) \geq -\alpha(h(\tau,\bx)), \ \forall \tau\in\Tc,
\end{align}
at every state $\bx$, the presence of infinite constraints in~\eqref{eq:cbf-conditions-infinite} complicates the control design process, the regularity conditions for invariance of $\mathcal{S}$, as well as the analysis of feasibility and safety guarantees of controllers derived from~\eqref{eq:cbf-conditions-infinite}.

In this paper, we seek to formalize CBF theory for safe sets described as in~\eqref{eq:safe-set-infinite-constraints}. In particular, we aim to establish:
\begin{enumerate}
    \item Necessary and sufficiency for forward invariance;
    \item A notion of CBF and optimal decay CBF (OD-CBF) \cite{ZengACC21,PioCDC25} tailored to~\eqref{eq:safe-set-infinite-constraints};
    \item Conditions under which the resulting CBF-based controllers enjoy desirable regularity properties and guarantee forward invariance of the set~\eqref{eq:safe-set-infinite-constraints};
    \item Conditions under which the control design process can be reduced to a finite number of CBF conditions, hence enabling the use of standard optimization solvers.
\end{enumerate}

\section{Invariance Under Infinite Constraints}\label{sec:infinite-invariance}
Motivated by the observation that forward invariance of a set defined by a single constraint~\eqref{eq:S} can be characterized by conditions imposed on its boundary, this section investigates the structure of the boundary for safe sets defined by infinitely many constraints. This characterization enables us to establish necessary and sufficient conditions for forward invariance by imposing barrier conditions directly on $\Sc$.

In the single-constraint case~\eqref{eq:S}, Nagumo's Theorem \cite{nagumo1942lage}, \cite[Ch. 4]{Blanchini}, \cite[Ch. 4]{AbrahamMarsdenRatiu} is widely cited to relate the condition:
\begin{equation}\label{eq:scalar-boundary}
     \pdv{h}{\bx}(\bx)\bF(\bx) \geq 0,~\forall \bx\in\partial\Sc,
\end{equation}
to invariance of $\mathcal{S}$. Yet \eqref{eq:scalar-boundary} is only meaningful for invariance when the boundary of this set is ``well-posed" in a certain sense. Here, 
one typically assumes that zero is a regular value of $h$, which is equivalent to assuming that for each $\bx\in\partial\mathcal{S}$ there exists a direction $\bv\in\real^n$ such that:
\begin{equation}\label{eq:regular-value}
    \pdv{h}{\bx}(\bx) \bv > 0.
\end{equation}
% for each $\bx\in\partial\mathcal S$, i.e., whenever $h(\bx)=0$. 
This implies there exists a direction along which $h$ is strictly increasing at the set boundary. Under this assumption, $\mathcal{S}$ can be rendered forward invariant under single integrator dynamics (full control of all states).
More importantly, this allows to establish that, for the original system \eqref{eq:closed-loop}, $\mathcal{S}$ is forward invariant if and only if \eqref{eq:scalar-boundary} holds \cite[Ex. 4.1.29]{AbrahamMarsdenRatiu}.

We now extend these conditions to 
% the case where
$\mathcal{S}$ 
% is 
defined by infinitely many constraints~\eqref{eq:safe-set-infinite-constraints}. An important object in this setting is the set of parameters corresponding to the active constraints:
\begin{equation}
    {\rm Act}(\bx) = \{\tau \in {\cal T} \;|\; h(\tau,\bx)=0\},
\end{equation}
which allows us to define the boundary of $\mathcal{S}$ as:
\begin{equation}
    \partial\mathcal{S} = \{\bx\in\mathcal{S}\mid {\rm Act}(\bx) \neq \emptyset\}.
\end{equation}
The following assumption extends \eqref{eq:regular-value} to infinite constraints.

\begin{assumption}[Well-posedness]\label{assump:well-posed}
    For each $\bx\in\partial\mathcal{S}$, there exists $\bv\in\R^n$ such that:
    \begin{equation}\label{eq:well-posed}
        \pdv{h}{\bx}(\tau,\bx)  \bv > 0,\quad \forall \tau\in{\rm Act}(\bx).
    \end{equation}
\end{assumption}

Assumption \ref{assump:well-posed} is similar to the
practical set assumption from \cite[Ch. 4]{Blanchini}, and leveraged in \cite{GurrietICCPS18,gurriet2020scalable}. To derive invariance conditions we follow the proof strategy in~\cite[Ch.~4.2.1]{Blanchini}, which establishes connections between forward invariance and boundary inequalities for sets defined by finitely many constraints. There, the proof requires the existence of a locally Lipschitz $\bx \mapsto \bphi(\bx)$ satisfying~\eqref{eq:well-posed}, in contrast to the pointwise existence of 
% a direction 
$\bv$ required in Assumption~\ref{assump:well-posed}. To bridge this gap, we provide the following lemma, illustrating that such a function $\bphi$ can be constructed under Assumption~\ref{assump:well-posed}.

\begin{lemma}[Smooth ascent direction on boundary]
\label{lem:practical}
    Consider the set $\Sc$ defined by an infinite number of inequality constraints as in~\eqref{eq:safe-set-infinite-constraints} with a function $h$ continuous in $\tau$ and continuously differentiable in $\bx$. If $\cal T$ is compact and Assumption~\ref{assump:well-posed} holds, then there exists a $C^\infty$ function $\bphi:\R^n \to \R^n$ such that for each $\bx\in\partial\Sc$:
    \begin{equation}\label{eq:smooth_phi}
    \pdv{h}{\bx}(\tau,\bx)  \bphi(\bx) > 0,
    \end{equation}
    for all $\tau \in {\rm Act}_{\epsilon(\bx)}(\bx)$ where 
    \begin{equation}\label{eq:Act-eps}
        {\rm Act}_{\epsilon(\bx)}(\bx) = \{\tau \in {\cal T} \;|\; h(\tau,\bx)\leq\epsilon(\bx)\},
    \end{equation}
    with a sufficiently small positive function $\epsilon : \R^n \to \R_{>0}$.
\end{lemma} 
\begin{proof}
    Denote $\bx\mapsto\bv(\bx)$ as the corresponding control at each $\bx\in\partial \Sc$ given by Assumption~\ref{assump:well-posed}. For each $\bx\in\partial \Sc$, let: $$\epsilon_1(\bx) \coloneqq \min_{\tau\in{\rm Act}(\bx)} \pdv{h}{\bx}(\tau,\bx) \bv(\bx) > 0,$$ be the minimum value achieved over the compact set ${\rm Act}(\bx)\subseteq {\cal T}$, which exists due to continuity of $\pdv{h}{\bx}$ in $\tau$. In addition, since $\cal T$ is compact, the continuity of the function with respect to $\tau$ is uniform, and therefore, there exists a sufficiently small open neighborhood $\Nc \supset {\rm Act}(\bx)$ such that:
    $$
    \pdv{h}{\bx}(\tau,\bx)  \bv(\bx) > \epsilon_1(\bx)/2 > 0
    $$
    for all $\tau \in \Nc$. Noting that $\Nc$ is open, the set difference ${\cal T} \setminus \Nc$ is compact. Then recalling that $h(\tau,\bx)\geq 0$ for all $\tau\in{\cal T}$ at $\bx\in\partial\Sc$ and $h(\tau,\bx)\neq 0$ for $\tau \not\in\Nc$,  the minimum value  $\epsilon_2(\bx) :=\min_{{\cal T} \setminus \Nc} h(\tau,\bx)>0$ is achieved with a strictly positive value. Consequently, the inequality above also holds for all $\tau$ in the subset ${\rm Act}_{\epsilon_2(\bx)}(\bx)\subseteq \Nc$. 

    Since $\pdv{h}{\bx}$ is uniformly continuous in $\tau$ and continuous in $\bx$, there exists a small neighborhood ${\cal W}(\bx)$ of $\bx$ such that:
    $$
    \pdv{h}{\bx}(\tau,\bx')  \bv(\bx) > 0,
    $$
    holds for all $\tau\in{\rm Act}_{\epsilon_2(\bx)}(\bx)$ at each $\bx'\in {\cal W}(\bx)$, with the input $\bv(\bx)$ given at $\bx$. We then note that for $\tau\not \in {\rm Act}_{\epsilon_2(\bx)}(\bx)$, i.e., $h(\tau,\bx)> \epsilon_2(\bx)$, there exists a sufficiently small neighborhood ${\cal W}_0(\bx)$ of $\bx$ such that $\tau\not \in {\rm Act}(\bx')$, i.e., $h(\tau,\bx')\neq 0$, for all $\bx'\in{\cal W}_0(\bx)$. In other words,  the neighborhood ${\cal W}_0(\bx)$ is such that  ${\rm Act}(\bx') \subseteq {\rm Act}_{\epsilon_2(\bx)}(\bx)$ for all $\bx' \in {\cal W}_0(\bx)$ by contrapositive. Thus, by selecting ${\cal W}(\bx)\subset \mathcal{W}_0(\mathbf{x})$ small enough, the inequality above holds for all $\tau\in{\rm Act}(\bx')$ at each $\bx'\in {\cal W}(\bx)$. With this newly defined neighborhood ${\cal W}(\bx)$, we next construct a $C^\infty$ function $\bphi$ from the pointwise-defined $\bv(\bx)$.

    The collection $\{{\cal W}(\bx)\}_{\bx\in\partial\Sc}$ is an open cover for a neighborhood ${\cal Q} \supset\partial {\cal S}$ of the boundary of the set $\cal S$. Since ${\cal Q}\subset \R^n$ is embedded in a Euclidean space, which is a differentiable manifold, there exists a countable $C^\infty$ partition of unity $\{\psi_j\}_{j\in\mathbb{N}}$, corresponding to a countable collection of sets $\{{\cal W}(\bx_j)\}_{j\in\mathbb{N}}$, subordinate to the cover, cf.~\cite[Theorem 1.11]{warner89}. The function $\bphi(\bx)= \sum_{j\in\mathbb{N}} \psi_j(\bx)\bv(\bx_j)$ defined from these partition of unity is also $C^\infty$ function. In addition, from convexity (linearity) of the function $\pdv{h}{\bx}(\tau,\bx) \cdot \bv$ in $\bv$, the constructed function $\bphi$ satisfies \eqref{eq:smooth_phi} at each $\bx\in\partial{\cal S}$ for all $\tau \in {\rm Act}(\bx)$. Similar to the argument at the beginning of the proof, from compactness of $\cal T$, we conclude that inequality~\eqref{eq:smooth_phi} holds for all $\tau \in {\rm Act}_{\epsilon(\bx)}(\bx)$ for a sufficiently small positive function $\epsilon : \R^n \to \R_{>0}$. 
\end{proof}

With Lemma~\ref{lem:practical}, Assumption~\ref{assump:well-posed} allows us to meaningfully characterize forward invariance of $\Sc$ through Lie-derivative-type conditions imposed on its boundary, extending the single-constraint characterization~\eqref{eq:scalar-boundary} to the infinite-constraint setting. We make this precise in the following theorem.

\begin{theorem}[Forward Invariance]\label{thm:infinite-nagumo}
    Consider the autonomous system~\eqref{eq:closed-loop} and the set $\Sc$ defined with an infinite number of inequality constraints as in~\eqref{eq:safe-set-infinite-constraints}, where $h$ is continuous in $\tau$ and continuously differentiable in $\bx$. Suppose $\cal T$ is compact, Assumption~\ref{assump:well-posed} holds, and $\bF$ is locally Lipschitz. Then, $\cal S$ is forward invariant for~\eqref{eq:closed-loop} if and only if, for each $\bx\in\partial{\cal S}$:
    \begin{equation}\label{eq:infinite-invariance-condition}
        \pdv{h}{\bx}(\tau,\bx)\bF(\bx) \geq 0,~\forall \tau\in{\rm Act}(\bx).
    \end{equation}
\end{theorem}
\begin{proof}
We follow a similar argument to that of \cite[Ch. 4.2.1]{Blanchini}.
\textbf{Sufficiency.} Consider the perturbed system:
\begin{equation}\label{eq:perturbed-dyn}
    \dot{\bx} = \bF_{\delta}(\bx) \coloneqq \bF(\bx) + \delta \bphi(\bx),
\end{equation}
where $\delta>0$ and $\bphi\,:\,\R^n\rightarrow\R^n$ is a $C^{\infty}$ vector field satisfying \eqref{eq:smooth_phi}, whose existence is guaranteed by Lemma~\ref{lem:practical}. Since $\bF$ is locally Lipschitz and $\bphi$ is smooth, $\bF_{\delta}$ is locally Lipschitz and admits a unique solution $t\mapsto\bx_{\delta}(t)$ from each initial condition $\bx_0$ defined on some maximal interval of existence. We will first show that $\mathcal{S}$ is forward invariant for the perturbed system \eqref{eq:perturbed-dyn}. 
To do so, let $\bx_0\in\mathcal{S}$ and suppose there exists a time $t_1\geq0$ such that $\bx_{\delta}(t_1)\in\partial\mathcal{S}$; if such a $t_1$ does not exist, then $\bx_{\delta}(t)\in\mathcal{S}$ for all $t$ since continuous solutions can only leave $\mathcal{S}$ through the boundary.
At $\bx_1\coloneqq \bx_{\delta}(t_1)\in\partial \mathcal{S}$, it follows from \eqref{eq:smooth_phi} and \eqref{eq:infinite-invariance-condition} that:
\begin{equation}
    \begin{aligned}
        \dot{h}(\tau,\bx_1) = & \pdv{h}{\bx}(\tau,\bx_1)\bF_{\delta}(\bx_1) \\ 
        =& \underbrace{\pdv{h}{\bx}(\tau,\bx_1)\bF(\bx_1)}_{\geq 0} + 
        \underbrace{\delta\pdv{h}{\bx}(\tau,\bx_1)\bphi(\bx_1)}_{>0} > 0,
    \end{aligned}
\end{equation}
for each 
$\tau\in\mathrm{Act}(\bx_1)$.
Since $\tau\mapsto\dot{h}(\tau,\cdot)$ is continuous and $\mathrm{Act}(\bx_1)\subseteq\mathcal{T}$ is compact, the continuity of $\dot{h}$ with respect to $\tau$ is uniform and there exists $\epsilon_1>0$ such that:
\begin{equation}
    \dot{h}(\tau,\bx_1) = \dot{h}(\tau,\bx_{\delta}(t_1)) \geq \epsilon_1 > 0,~\forall \tau\in\mathrm{Act}(\bx_1).
\end{equation}
Further, since $t\mapsto \dot{h}(\tau,\bx_{\delta}(t))$ is continuous, and the continuity of $\dot{h}$ with respect to $\tau$ is uniform, there exists $c_1>0$ such that for $t\in[t_1,t_1 + c_1)$:
\begin{equation}
    \dot{h}(\tau,\bx_{\delta}(t)) \geq \frac{\epsilon_1}{2} > 0,~\forall \tau\in\mathrm{Act}(\bx_1).
\end{equation}
Using the same argument as in the proof of Lemma \ref{lem:practical}, we construct a neighborhood $\mathcal{E}\supset\mathrm{Act}(\bx_1)$ and constant $\epsilon_2\coloneqq\min_{\mathcal{T}\setminus\mathcal{E}}h(\tau,\bx_1)>0$ such that for $t\in[t_1,t_1 + c_1)$:
\begin{equation}
    \dot{h}(\tau,\bx_{\delta}(t)) \geq \frac{\epsilon_1}{4} > 0,~\forall \tau\in\mathrm{Act}_{\epsilon_2}(\bx_1).
\end{equation}
Hence, $t\mapsto h(\tau,\bx_{\delta}(t))$ is strictly increasing, uniformly in $\tau\in\mathrm{Act}_{\epsilon_2}(\bx_1)$, on $[t_1,t_1 + c_1)$, which implies that:
\begin{equation}
    h(\tau,\bx_{\delta}(t))\geq h(\tau,\bx_{\delta}(t_1)) = 0,~\forall \tau\in\mathrm{Act}_{\epsilon_2}(\bx_1),
\end{equation}
for $t\in[t_1,t_1 + c_1)$.
We now show that $h(\tau,\bx_{\delta}(t))\geq0$ for $\tau\in\mathcal{T}\setminus \mathrm{Act}_{\epsilon_2}(\bx_1)$ on a right interval of $t_1$. From \eqref{eq:Act-eps}:
\begin{equation}
    h(\tau,\bx_1) > \epsilon_2 > 0,\quad \forall \tau\in\mathcal{T}\setminus \mathrm{Act}_{\epsilon_2}(\bx_1).
\end{equation}
Using the uniform continuity of $h$ 
with respect to $\tau\in\mathcal{T}\setminus \mathrm{Act}_{\epsilon_2}(\bx_1)$
and the strictness of the above inequality, it follows that there exists $c_2>0$ such that for each $\tau\in\mathcal{T}\setminus \mathrm{Act}_{\epsilon_2}(\bx_1)$ we have:
\begin{equation}
    h(\tau,\bx_{\delta}(t)) > \frac{\epsilon_2}{2} > 0,~ \forall t\in[t_1, t_1 + c_2).
\end{equation}
Taking $c\coloneqq \min\{c_1,c_2\}$, we thus have that for all $\tau\in\mathcal{T}$:
\begin{equation}
    h(\tau,\bx_{\delta}(t)) \geq0,~ \forall t \in [t_1, t_1 + c),
\end{equation}
thereby showing that $\bx_{\delta}(t)\in\mathcal{S}$ for all $t\in[t_1,t_1 + c)$.
Repeating this argument for any time $t$ at which $\bx_{\delta}(t)\in\partial\mathcal{S}$ shows that solutions must always remain in $\mathcal{S}$ for a nonempty interval of time after reaching the boundary, implying that $\bx_{\delta}(t)\in\mathcal{S}$ for all $t$ in the interval of existence provided that $\bx_0\in\mathcal{S}$, i.e., $\mathcal{S}$ is forward invariant for the perturbed system \eqref{eq:perturbed-dyn} for any $\delta>0$.
The remainder of the sufficiency proof follows the same argument as that of \cite[Ch. 4.2.1]{Blanchini} by leveraging the fact that solutions of the perturbed system converge uniformly to those of the unperturbed system on any compact time interval \cite[Thm. 3.5]{Khalil}.

\textbf{Necessity.} 
We now show that if $\mathcal{S}$ is forward invariant for $\bF$, then \eqref{eq:infinite-invariance-condition} must hold. Pick an initial condition $\bx_0\in\partial \mathcal{S}$. Since $\mathcal{S}$ is forward invariant, we have $\bx(t)\in\mathcal{S}$ for all $t \in \mathcal{I}(\bx_0)$, which implies that for all $t \in \mathcal{I}(\bx_0)$ we have $h(\tau,\bx(t)) \geq 0$ for all $\tau\in\mathcal{T}$. Pick any $\tau^*\in\mathrm{Act}(\bx_0)$, define $\rho(t)=h(\tau^*,\bx(t))$, and note that:
\begin{equation*}
    \dot{\rho}(t) = \pdv{h}{\bx}(\tau^*,\bx(t))\bF(\bx(t))
\end{equation*}
Now, since $\rho(0)=0$ and $\rho(t)\geq0$ for all $t\in \mathcal{I}(\bx_0)$ we must have $\dot{\rho}(0)\geq 0$, implying that:
\begin{equation*}
    \pdv{h}{\bx}(\tau^*,\bx_0)\bF(\bx_0) \geq 0.
\end{equation*}
Since the choices of $\bx_0\in\partial\mathcal{S}$ and $\tau^*\in \mathrm{Act}(\bx_0)$ were arbitrary, the above holds for any $\bx_0\in\partial\mathcal{S}$ and any $\tau \in \mathrm{Act}(\bx_0)$, implying that \eqref{eq:infinite-invariance-condition} holds, as desired. 
\end{proof}

Theorem \ref{thm:infinite-nagumo} extends invariance conditions based on Nagumo's Theorem \cite[Ch. 4]{Blanchini} to the infinite constraint setting. We now bridge the gap between these necessary and sufficient conditions for set invariance and barrier functions.

\begin{corollary}
    Let the assumptions of Theorem \ref{thm:infinite-nagumo} hold. If there exists a class $\mathcal{K}$ function $\alpha$ such that for each $\bx\in\mathcal{S}$:
    \begin{equation}\label{eq:infinite-barrier}
        \pdv{h}{\bx}(\tau,\bx)\bF(\bx) \geq - \alpha(h(\tau, \bx)),\quad \forall \tau\in\mathcal{T},
    \end{equation}
    then $\mathcal{S}$ is forward invariant. 
\end{corollary}

\begin{proof}
    When $\bx\in\partial \mathcal{S}$, we have $\mathrm{Act}(\bx)\neq\emptyset$ and $h(\tau,\bx)=0$ for each $\tau\in \mathrm{Act}(\bx)$. Since $\mathrm{Act}(\bx)\subset\mathcal{T}$, \eqref{eq:infinite-barrier} also holds for any $\tau\in\mathrm{Act}(\bx)$, so, when $\bx\in\partial\mathcal{S}$, \eqref{eq:infinite-barrier} reduces to \eqref{eq:infinite-invariance-condition}, implying forward invariance of $\mathcal{S}$ by Theorem \ref{thm:infinite-nagumo}.
\end{proof}

The existence of $\alpha\in\mathcal{K}$ satisfying \eqref{eq:infinite-barrier} is sufficient, but not necessary for forward invariance. This gap can be closed using the idea of an optimal decay barrier function \cite{ZengACC21,PioCDC25}.

\begin{definition}\label{def:OD-BF}
    A function $h\,:\,\R\times\R^n\rightarrow\R$, $h(\tau,\bx)$, continuous in $\tau$ and continuously differentiable in $\bx$, is said to be 
    an \emph{optimal-decay barrier function} (OD-BF) for \eqref{eq:closed-loop} on the set $\mathcal{S}$ as in \eqref{eq:safe-set-infinite-constraints} if for any $\alpha\in\mathcal{K}$ there exists a function $\theta\,:\,\mathcal{T} \times \R^n\rightarrow\R_{\geq0}$ such that for each $\bx\in \mathcal{S}$:
    \begin{equation}\label{eq:infinite-OD-BF}
        \pdv{h}{\bx}(\tau,\bx)\bF(\bx)\geq -\theta(\tau,\bx)\alpha(h(\tau,\bx)),\quad \forall \tau \in \mathcal{T}.
    \end{equation}
\end{definition}

The following result shows that under the assumptions of Theorem \ref{thm:infinite-nagumo}, the existence of an OD-BF is necessary and sufficient for set invariance.

\begin{proposition}\label{prop:infinite-OD-BF}
    Consider system \eqref{eq:closed-loop} and a set $\mathcal{S}$ as in \eqref{eq:safe-set-infinite-constraints} defined by a function $h\,:\,\R\times\R^n\rightarrow\R$.
    Suppose that $\bF$ is locally Lipschitz, $\mathcal{T}$ is compact, and Assumption \ref{assump:well-posed} holds. Then $\mathcal{S}$ is forward invariant if and only if $h$ is an OD-BF.
\end{proposition}

\begin{proof}
    The proof follows similar steps to that of \cite[Thm. 2]{PioCDC25}.
    We will show that \eqref{eq:infinite-invariance-condition} and \eqref{eq:infinite-OD-BF} are equivalent. To show that \eqref{eq:infinite-OD-BF} $\implies$ \eqref{eq:infinite-invariance-condition} note that since $\mathrm{Act}(\bx)\subseteq\mathcal{T}$ it follows from \eqref{eq:infinite-OD-BF} that when $\bx\in\partial\mathcal{S}$ we have $\pdv{h}{\bx}(\tau,\bx)\bF(\bx)\geq 0$ for all $\tau\in \mathcal{T}$, which implies that \eqref{eq:infinite-invariance-condition} holds. To show that \eqref{eq:infinite-invariance-condition} $\implies$ \eqref{eq:infinite-OD-BF}, we will construct a $\theta\,:\,\mathcal{T} \times \R^n\rightarrow\R_{\geq 0}$ satisfying \eqref{eq:infinite-OD-BF}. Let $\alpha\in \mathcal{K}$ be arbitrary and define:
    \begin{equation}\label{eq:equivalent-theta}
        \theta(\tau,\bx) \coloneqq \begin{cases}
            0 & h(\tau,\bx) = 0, \\
            \frac{\mathrm{ReLU}(-\pdv{h}{\bx}(\tau,\bx)\bF(\bx))}{\alpha(h(\tau,\bx))} & h(\tau,\bx) \neq 0,
        \end{cases}
    \end{equation}
    where $\mathrm{ReLU}(\cdot)=\max\{0,\cdot\}$, which satisfies $\theta(\tau,\bx)\geq0$ for all $(\tau,\bx)\in\mathcal{T}\times\mathcal{S}$. 
    %Now take $\theta(\bx)\coloneqq \sup_{\tau\in \mathcal{T}}\vartheta(\tau,\bx)$. 
    When $h(\tau,\bx)=0$, i.e., when $\bx\in\partial \mathcal{S}$, \eqref{eq:infinite-OD-BF} with $\theta$ as above is equivalent to \eqref{eq:infinite-invariance-condition}. Otherwise, when $h(\tau,\bx)\neq0$ we have:
    \begin{equation*}
        \begin{aligned}
            \pdv{h}{\bx}(\tau,\bx)\bF(\bx) \geq & -\mathrm{ReLU}\left(-\pdv{h}{\bx}(\tau,\bx)\bF(\bx))\right)\\ 
            %& -\vartheta(\tau,\bx)\alpha(h(\tau,\bx)) \\ 
            \geq & -\theta(\tau,\bx)\alpha(h(\tau,\bx))
        \end{aligned}
    \end{equation*}
    implying \eqref{eq:infinite-OD-BF} holds. Since \eqref{eq:infinite-invariance-condition} and \eqref{eq:infinite-OD-BF} are equivalent, and, under the hypotheses of the proposition statement \eqref{eq:infinite-invariance-condition} is a necessary and sufficient condition for the forward invariance of $\mathcal{S}$, \eqref{eq:infinite-OD-BF} is also a necessary and sufficient condition for forward invariance of $\mathcal{S}$, as desired. 
\end{proof}

% One of the limitations of the preceding result is that $\theta$ depends on $\tau$. In the context of OD-CBFs, this would result in an infinite number of decision variables, turning a finite-dimensional optimization problem, i.e, a QP, into an infinite-dimensional one.
% Under stronger assumptions, the optimal decay variable $\theta$ from \eqref{eq:infinite-OD-BF} can be made independent of $\tau$.

% \begin{corollary}[Reduction to a single OD variable]\label{cor:strictness-thetax}
%     Consider the same hypothesis as in Proposition~\ref{prop:infinite-OD-BF} and suppose that:
%     \begin{align}\label{eq:strictness-assumption}
%         \frac{\partial h}{\partial \bx}(\tau,\bx) \bF(\bx) > 0, \ \forall \tau \in \mathrm{Act}(\bx),
%     \end{align}
%     for all $\bx\in\partial\Sc$.
%     Then, there exists $\alpha\in\Kc$ and $\bar{\theta}:\real^n\to\real_{\geq0}$ such that for each $\bx\in\Sc$:
%     \begin{equation}\label{eq:infinite-OD-BF-thetax}
%         \pdv{h}{\bx}(\tau,\bx)\bF(\bx)\geq -\bar{\theta}(\bx)\alpha(h(\tau,\bx)),\quad \forall \tau \in \mathcal{T}.
%     \end{equation}
% \end{corollary}
% \begin{proof}
%     We note that under the strictness assumption~\eqref{eq:strictness-assumption},
%     the function $\theta(\tau,\bx)$ constructed in the proof of Proposition~\ref{prop:infinite-OD-BF} is continuous. 
%     Hence, the function $\bar{\theta}(\bx) \coloneqq \max_{\tau\in\mathcal{T}}\theta(\tau,\bx)$ is well-defined and by a similar argument to the one in the proof of Proposition~\ref{prop:infinite-OD-BF}, it satisfies~\eqref{eq:infinite-OD-BF-thetax}.
% \end{proof}

\begin{remark}
    The results of this section for sets defined by an infinite number of constraints also capture sets defined by a finite number of constraints as a special case. Our main hypothesis is that $(\tau,\bx)\mapsto h(\tau,\bx)$ is continuous in $\tau\in\mathcal{T}$ for compact $\mathcal{T}$ and continuously differentiable in $\bx\in\R^n$. 
    % While our results are motivated by the setting where $\mathcal{T}\subset\R$, such ideas also specialize to the setting where $\mathcal{T}=\{1,2,\dots,N\}$ is a finite set that, e.g., labels a finite number of inequality constraints $h_\tau(\bx)\geq0$ for $\tau\in \mathcal{T}$, since $\tau\mapsto h(\tau,\bx)$ is trivially continuous over a finite set $\mathcal{T}$.
\end{remark}

\section{Control Under Infinite Constraints}
In this section, we extend the results from the previous section to control systems of the form~\eqref{eq:control-affine}. 
We start by defining an analogue of Definition~\ref{def:OD-BF} for such control systems.

\begin{definition}\label{def:OD-CBF}
    A function $h:\real\times\real^n\to\real$, $h(\tau,\bx)$, continuous in $\tau$ and continuously differentiable in $\bx$, is said to be an \textit{optimal-decay control barrier function} (OD-CBF) for~\eqref{eq:control-affine} on the set $\Sc$ as in~\eqref{eq:safe-set-infinite-constraints} if for any $\alpha\in\Kc^e$:
    \begin{equation}\label{eq:od-cbf-condition}
        \sup\limits_{\substack{\bu\in\real^m\\ \omega\in\real_{\geq0}} } 
        \Big\{
        \frac{\partial h}{\partial \bx}(\tau,\bx)( \bf(\bx) + \bg(\bx)\bu ) + \omega \alpha( h(\tau,\bx) ) 
        \Big\}
        \geq 0,
    \end{equation}
    for each $\bx\in\Sc$, $\tau\in\Tc$.
\end{definition}
An important feature of Definition~\ref{def:OD-CBF} is that it uses a single OD variable $\omega$ despite having infinite constraints.
% This is particularly important because 
If Definition~\ref{def:OD-CBF} required a different OD variable for each
constraint, optimization-based designs including the OD variables as decision variables would be impracticable, since they would require infinite decision variables.
Instead, Definition~\ref{def:OD-CBF} allows for the use of control synthesis meeting \eqref{eq:od-cbf-condition} via the semi-infinite optimization problem: 
% (which is a QP if $\Uc$ is polyhedral):
\begin{equation}\label{eq:infinite-OD-CBF-QP}
    \begin{aligned}
        \begin{bmatrix}
            \bk(\bx)\\ \bar{\theta}(\bx)
        \end{bmatrix} 
        & = \argmin_{\substack{\bu\in\real^m \\ \omega\in\R}}  \quad \tfrac{1}{2}\|\bu - \bk_{\rm{d}}(\bx)\|^2 + \tfrac{1}{2}p(\omega - \theta_{\rm{d}})^2 \\ 
        \mathrm{s.t.} & \quad \dot{h}(\tau,\bx,\bu) \geq -\omega \alpha( h(\tau,\bx) ), \ \forall \tau\in\Tc \\
        & \quad \omega \geq 0,
    \end{aligned}
\end{equation}
where $\dot{h}(\tau,\bx,\bu) = \frac{\partial h}{\partial \bx}(\tau,\bx)( \bf(\bx) + \bg(\bx)\bu )$. The following result ensures that if the barrier conditions are strictly feasible on the boundary of the safe set, then the OD-CBF condition~\eqref{eq:od-cbf-condition} is guaranteed to be feasible.

\begin{proposition}[Reduction to a single OD variable]\label{prop:reduction-one-OD-variable}
    Suppose that $\Tc$ is compact, Assumption~\ref{assump:well-posed} holds, and for all $\bx\in\partial\Sc$, 
    there exists $\bu\in\real^m$ such that: 
    \begin{align}\label{eq:strictness-assumption}
        \frac{\partial h}{\partial \bx}(\tau,\bx)(\bf(\bx) + \bg(\bx)\bu) > 0, \ \forall \tau\in\mathrm{Act}(\bx).
    \end{align}
    Then, for any $\alpha\in\Kc^e$, 
    there exists a neighborhood $\Dc$ of $\Sc$ such that~\eqref{eq:infinite-OD-CBF-QP} is strictly feasible for all $\bx\in\Dc$.
\end{proposition}
\begin{proof}
    By an argument analogous to Lemma~\ref{lem:practical}, there exists a smooth controller $\bar{\bu}:\real^n\to\real^m$ 
    and a sufficiently small positive function $\epsilon:\real^n\to\real_{>0}$ 
    such that for all $\bx\in\partial\Sc$:
    \begin{align}\label{eq:strictness-condition}
        \frac{\partial h}{\partial \bx}(\tau,\bx)(\bf(\bx)+\bg(\bx)\bar{\bu}(\bx) ) > 0, \ \forall\tau\in\mathrm{Act}_{\epsilon(\bx)}(\bx),
    \end{align}    
    with $\text{Act}_{\epsilon(\bx)}(\bx)$ defined as in~\eqref{eq:Act-eps}.
    Now, let $\bF(\bx) := \bf(\bx) + \bg(\bx)\bar{\bu}(\bx)$ and define:
    \begin{align*}
        \theta(\tau,\bx) = \frac{ | \pdv{h}{\bx}(\tau,\bx) \bF(\bx) | + 1 }{\alpha( \epsilon(\bx) )}
    \end{align*}
    % for each $\tau\in\Tc$, define \begin{equation}\label{eq:equivalent-theta-Prop2}
    %     \theta(\tau,\bx) \coloneqq \begin{cases}
    %         0 & \tau \in \text{Act}_{\epsilon(\bx)}(\bx), \\
    %         \frac{ \| \pdv{h}{\bx}(\tau,\bx) \|
    %         \| \bF(\bx) \| + \alpha(\epsilon(\bx)) }{\alpha( \epsilon(\bx) )} & \tau\in\Tc\setminus\text{Act}_{\epsilon(\bx)}(\bx),
    %     \end{cases}
    % \end{equation}
    %We need to add the extra $\alpha(\epsilon(\bx))$ in the numerator just in case 
    %\| \pdv{h}{\bx}(\tau,\bx) \| \| \bF(\bx) \| = 0.
    Further let $\bar{\theta}_1(\bx) = \sup_{\tau\in\Tc} \theta(\tau,\bx)$,
    which is finite because $\Tc$ is compact and $\frac{\partial h}{\partial \bx}(\tau,\bx)$ is continuous.
    Now, for each $\bx\in\partial\Sc$ and $\tau\in\Tc$ we have: 
    \begin{align*}
        \frac{\partial h}{\partial \bx}(\tau,\bx)(\bf(\bx)+\bg(\bx)\bar{\bu}(\bx) ) + \bar{\theta}(\bx)\alpha(h(\tau,\bx)) > 0.
    \end{align*}
    In particular, the inequality above holds for $\tau\in\text{Act}_{\epsilon(\bx)}(\bx)$ due to condition~\eqref{eq:strictness-condition}, together with the fact that $\bar\theta(\bx)$ and $\alpha(h(\tau,\bx))$ are both nonnegative. At the same time, for $\tau\in\Tc\setminus \text{Act}_{\epsilon(\bx)}(\bx)$, we have $h(\tau,\bx)> \epsilon$, and therefore, with  $\bar\theta_1(\bx)$, the left hand side is lower-bounded by $1$. 
    Hence,~\eqref{eq:infinite-OD-CBF-QP} is strictly feasible for all $\bx\in\partial\Sc$ and $\tau\in\Tc$.
    Then, by an argument analogous to the one in the proof of Lemma~\ref{lem:practical}, we can construct a smooth function $\tilde{\theta}:\real^n\to\real_{>0}$ such that: 
    \begin{align}\label{eq:strictly-feasible-expression}
        \frac{\partial h}{\partial \bx}(\tau,\bx)(\bf(\bx)+\bg(\bx)\bar{\bu}(\bx) ) + \tilde{\theta}(\bx)\alpha(h(\tau,\bx)) > 0,
    \end{align}
    holds for all $\bx\in\partial\Sc$ and $\tau\in\Tc$.
    Since the left hand side of~\eqref{eq:strictly-feasible-expression} is smooth in $\bx$ and $\tau$,
    there exists a neighborhood $\Nc$ of $\partial\Sc$ for which~\eqref{eq:infinite-OD-CBF-QP} is strictly feasible (by taking $\bu = \bar{\bu}(\bx)$ and $\omega = \tilde{\theta}(\bx)$).
    % Now, for $\bx^{\prime}\in\real^n$, define 
    % \begin{align*}
    %     &m_{\bx}(\bx^\prime) = \\
    %     &\min\limits_{\tau\in\Tc}
    %     \frac{\partial h}{\partial \bx}(\tau,\bx^\prime)(\bf(\bx^\prime)+\bg(\bx^\prime)\bar{\bu}(\bx) ) + \bar{\theta}_1(\bx)\alpha(h(\tau,\bx^\prime))
    % \end{align*}
    % By Berge's Maximum Theorem~\cite[Thm. 16.31]{Aliprantis06}, $m_{\bx}$ is continuous.
    % Since $m_{\bx}(\bx) > 0$, there exists a neighborhood $\mathcal{N}(\bx)$ of $\bx$ for which~\eqref{eq:infinite-OD-CBF-QP} is strictly feasible.
    % This means that there exists a neighborhood $\Nc$ of $\partial\Sc$ for which~\eqref{eq:infinite-OD-CBF-QP} is strictly feasible.
    It remains to show that~\eqref{eq:infinite-OD-CBF-QP} is strictly feasible for $\bx\in\Sc\setminus\Nc$. Given any such $\bx$, taking $\bu = \bar{\bu}(\bx)$ and:
    $$
    \omega = \frac{ \max_{\tau\in\Tc} | \pdv{h}{\bx}(\tau,\bx)  (\bf(\bx) + \bg(\bx)\bar{\bu}(\bx) ) | + 1 }{\min_{\tau\in\Tc} \alpha( h(\tau,\bx) )},
    $$
    satisfies the constraints strictly, as desired.
    % Let $\bar{\epsilon}(\bx) := \min\limits_{\tau\in\Tc}h(\tau,\bx)$, and note that $\bar{\epsilon}(\bx) > 0$ for all $\bx\in\Sc\setminus\Nc$.
    % Further define 
    % \begin{align*}
    %     \theta_2(\tau,\bx) = \frac{ | \pdv{h}{\bx}(\tau,\bx)  \bF(\bx) | + 1 }{\alpha( \bar{\epsilon}(\bx) )},
    % \end{align*}
    % and $\bar{\theta}_2(\bx) = \sup\limits_{\tau\in\Tc} \theta_2(\tau,\bx)$.
    % Strict feasibility of~\eqref{eq:infinite-OD-CBF-QP} at $\bx\in\Sc\setminus\Nc$ follows by letting $\omega = \bar{\theta}_2(\bx)$.
    % Now the result follows by taking $\Dc = \Sc\cup\Nc$.
\end{proof}

% An immediate consequence of Proposition~\ref{prop:reduction-one-OD-variable} is that \eqref{eq:infinite-OD-CBF-QP} is feasible, as one can always take $\bu$ and $\bar{\theta}$ as in the proof to satisfy \eqref{eq:od-cbf-condition}, and thus the constraints in \eqref{eq:infinite-OD-CBF-QP}.
Although Proposition~\ref{prop:reduction-one-OD-variable} guarantees feasibility of 
% the optimization problem
\eqref{eq:infinite-OD-CBF-QP}, to ensure that the resulting closed-loop system is safe, we must design such functions to be sufficiently regular.
% Now, consider the OD-CBF-QP with infinite constraints
% \begin{equation}\label{eq:infinite-OD-CBF-QP}
%     \begin{aligned}
%         \begin{bmatrix}
%             \bk(\bx)\\\theta(\bx)
%         \end{bmatrix} &= \argmin_{\substack{\bu\in\mathcal{U} \\ \omega\in\R}}  \quad \tfrac{1}{2}\|\bu - \bk_{\rm{d}}(\bx)\|^2 + \tfrac{1}{2}p(\omega - \theta_{\rm{d}})^2 \\ 
%         \mathrm{s.t.} & \quad \frac{\partial h}{\partial \bx}(\tau,\bx)\bf(\bx) + 
%         \frac{\partial h}{\partial \bx}(\tau,\bx)\bg(\bx)\bu \geq -\omega \alpha(h(\tau,\bx))
%         \\
%         & \quad \omega \geq 0.
%     \end{aligned}
% \end{equation}
%
% As shown in Proposition~\ref{prop:reduction-one-OD-variable},~\eqref{eq:infinite-OD-CBF-QP} is guaranteed to be feasible.
% However, in order to guarantee that the resulting controller is safe, we need to establish it is sufficiently regular.
To do so, we consider the general setting of designing a controller to satisfy
an infinite set of parametric constraints:
\begin{align}\label{eq:infinite-inequalities}
        a(\tau,\bx) + \bb(\tau,\bx) \bnu \leq 0, \quad \tau \in \mathcal{T},
\end{align}
where $a:\mathcal{T}\times\mathcal{X}\to\real$,
$\bb:\mathcal{T}\times\mathcal{X}\to\real^{1\times m_{\bnu}}$ are continuous.
% Such constraints could, e.g., be induced by CBF-like conditions on the set from \eqref{eq:safe-set-infinite-constraints} where the input must satisfy:
% \begin{equation}
%     \underbrace{-
%     \pdv{h}{\bx}(\tau,\bx)\bf(\bx) - \alpha(h(\tau,\bx))}_{a(\tau,\bx)} \underbrace{-\pdv{h}{\bx}(\tau,\bx)\bg(\bx)}_{\bb(\tau,\bx)}\bu \leq 0,
% \end{equation}
% for all $\tau\in\mathcal{T}$. 
OD-CBF constraints take the form~\eqref{eq:infinite-inequalities} by defining $\bnu=[\bu,\omega]^\top$.
% appending the control input with the additional OD variable.
To extend the ideas from Sec. \ref{sec:infinite-invariance} to control design, we first present an extension of Artstein's Theorem~\cite{Artstein83} to infinite constraints, outlining conditions under which there exists a smooth feedback controller satisfying \eqref{eq:infinite-inequalities}.

\begin{proposition}[Artstein for infinite constraints]\label{prop:artstein-infinite-constraints}
    Let $\mathcal{T}\subset\real$ be compact. 
    Let $\mathcal{X}\subset\real^n$ and
    consider continuous functions $a:\mathcal{T}\times\mathcal{X}\to\real$,
    $\bb:\mathcal{T}\times\mathcal{X}\to\real^{1\times m_{\bnu}}$ defining \eqref{eq:infinite-inequalities}.
    % Consider the set of inequalities 
    % \begin{align}\label{eq:infinite-inequalities}
    %     a(\tau,\bx) + \bb(\tau,\bx) \bu \leq 0, \quad \tau \in \mathcal{T}.
    % \end{align}
    Suppose that for each $\bx\in\mathcal{X}$, there exists $\bnu\in\real^{m_{\bnu}}$ satisfying:
    \begin{equation}\label{eq:slaters}
        a(\tau,\bx) + \bb(\tau,\bx) \bnu < 0, \quad \forall \tau \in \mathcal{T}.
    \end{equation}
    Then, there exists a $\mathcal{C}^{\infty}$ selection $\bk:\mathcal{X}\to\real^{m_{\bnu}}$ such that:
    \begin{equation}\label{eq:a-b-strictly}
        a(\tau,\bx) + \bb(\tau,\bx)\bk(\bx) < 0,\quad \forall (\tau,\bx)\in\mathcal{T}\times \mathcal{X}.
    \end{equation}
\end{proposition}
\begin{proof}
    For each $\bx\in\mathcal{X}$, let $\bv(\bx)$ be the vector $\bnu\in\real^{m_{\bnu}}$ satisfying~\eqref{eq:slaters}.
    Consider the function:
    \begin{align*}
        m_{\bx}(\bx^{\prime}) = \min\limits_{\tau\in \mathcal{T}} a(\tau,\bx^{\prime}) + \bb(\tau,\bx^{\prime}) \bv(\bx).
    \end{align*}
    Since $a$ and $\bb$ are continuous and $\mathcal{T}$ is compact, $m_{\bx}$ is continuous by Berge's Maximum Theorem~\cite[Thm. 16.31]{Aliprantis06}.
    Hence, there exists a neighborhood $\Wc(\bx)$ of $\bx$ such that $m_{\bx}(\bx^{\prime}) < 0$ for all $\bx^{\prime}\in \Wc(\bx)$.
    Now the proof follows the same Artstein-style argument as~\cite[Lemma 6.5]{PO-BC-LS-JC:23-auto}.
    % which we include here for completeness.
    % Note that $\{ W(\bx) \}_{\bx\in\mathcal{X}}$ is an open cover of $\mathcal{X}$.
    % Therefore, there exists a countable partition of unity $\{ \psi_j \}_{j\in J}$ (with $J$ a countable set)
    % subordinate to the cover~\cite[Theorem 1.11]{warner89}.
    % In other words, for each $j$, there exists $\bx_j\in\mathcal{X}$ such that $\text{supp}(\psi_j)$ is a subset of $W(\bx_j)$, each of which has an associated control $\bv(\bx_j) =: \bv_j$ satisfying~\eqref{eq:infinite-inequalities}. Now, define $\bk(\bx) = \sum_{j\in J} \psi_j(\bx) \bv_j$.
    % We have:
    % \begin{align*}
    %     a(\tau,\bx) + \bb(\tau,\bx)\bk(\bx) = a(\tau,\bx) + \sum_{j\in J} \bb(\tau,\bx)\bv_j \psi_j(\bx).
    % \end{align*}
    % For each $\bx\in\mathcal{X}$, let $J(\bx) = \{ j\in J : \psi_j(\bx) \neq 0 \}$.
    % Since $\text{supp}(\psi_j)$ is a subset of $W(\bx_j)$, $\bb(\tau,\bx) \bv_j < -a(\tau,\bx)$ for all $j\in J(\bx)$. The result follows from the fact that $\sum_{j\in J(\bx)} \psi_j(\bx) = 1$.
\end{proof}

Proposition \ref{prop:artstein-infinite-constraints} is non-constructive. It proves existence of a smooth controller, but does not provide a procedure to obtain it. In practice, controllers meeting such constraints are often synthesized using optimization-based techniques:
\begin{equation}\label{eq:infinite-selection}
    \bk(\bx) = \argmin_{\bnu\in\mathcal{V}(\bx)}\|\bnu-\bk_\des(\bx)\|^2,
\end{equation}
where $\bk_\des:\mathcal{X}\rightarrow\R^{m_{\bnu}}$ is a nominal controller and $\mathcal{V}:\mathcal{X}\rightrightarrows\R^{m_{\bnu}}$ is the set-valued map:
\begin{equation}
        \mathcal{V}(\bx)\!\coloneqq\! \{\bnu\in\real^{m_{\bnu}}\mid a(\tau,\bx) + \bb(\tau,\bx)\bnu \leq 0, \forall \tau\in\mathcal{T} \},
\end{equation}
% \begin{equation}\label{eq:general-semi-infinite-qp}
%     \begin{aligned}
%         \min_{\bu\in\R^m} \quad & \frac{1}{2}\|\bu - \bk_{\rm{d}}(\bx)\|^2 \\ 
%         \mathrm{s.t.} \quad &  a(\tau,\bx) + \bb(\tau,\bx) \bu \leq 0, \quad \forall\tau \in \mathcal{T}.
%     \end{aligned}
% \end{equation}
Asking for this controller to be locally Lipschitz is typically too strong of a request; however, under conditions similar to those in Proposition \ref{prop:artstein-infinite-constraints}, it is at least continuous.

\begin{proposition}[Continuity of optimization-based controller]
    % Let $J\,:\,\mathcal{X}\times\R^m\rightarrow\R$, $a\,:\,\mathcal{T}\times\mathcal{X}\rightarrow\R$, and $\bb\,:\,\mathcal{T}\times\mathcal{X}\rightarrow\R^{1\times m}$ be continuous and suppose that $\bu\mapsto J(\bx,\bu)$ is strictly convex and coercive for each $\bx\in\mathcal{X}$. Provided that for each $\bx\in\mathcal{X}$ there exists a $\bu\in\R^m$ satisfying \eqref{eq:slaters}, then $\bk\,:\,\mathcal{X}\rightarrow\R^m$ as in \eqref{eq:infinite-selection} is continuous.
    Consider the same setting as in Proposition~\ref{prop:artstein-infinite-constraints}. Provided that for each $\bx\in\mathcal{X}$ there exists a $\bnu\in\real^{m_{\bnu}}$ satisfying \eqref{eq:slaters}, then $\bk\,:\,\mathcal{X}\rightarrow\real^{m_{\bnu}}$ as in \eqref{eq:infinite-selection} is continuous.
\end{proposition}

\begin{proof}
    From \cite[Prop 2.19]{freeman1996robust}, the minimal selection on a finite dimensional space~$\real^{m_{\bnu}}$ as in \eqref{eq:infinite-selection} is continuous if the set-valued map $\Vc$ is lower semicontinuous~\cite[Def. 2.2]{freeman1996robust} on~$\Xc$. To establish this property, consider a fixed $\bx\in\Xc$ and any open set $\mathcal{O}\subset \real^{m_{\bnu}}$ such that $\Vc(\bx)\cap\mathcal{O}$ is nonempty. Because the assumption~\eqref{eq:slaters} ensures $\Vc(\bx)$ has a nonempty interior, its intersection with the open set $\mathcal{O}$ must also include a $\bnu\in\Vc(\bx)\cap\mathcal{O}$ in the interior of $\Vc(\bx)$. Given such a~$\bnu$, it must satisfy~\eqref{eq:slaters}, and therefore, $\bnu\in\Vc(\bx')$ for $\bx'$ in a neighborhood $\Wc(\bx)$ of $\bx$, as shown in the proof of Proposition~\ref{prop:artstein-infinite-constraints}. As such, $\bnu\in\Vc(\bx)\cap\mathcal{O}\subseteq\mathcal{O}$ also belongs to $\Vc(\bx')\cap \mathcal{O}$, proving lower semicontinuity of $\Vc$ at each $\bx\in\Xc$ and continuity of $\bk$ on $\Xc$, as needed.
\end{proof}

While the previous result guarantees that controllers synthesized from \eqref{eq:infinite-selection} are continuous, the results in Sec. \ref{sec:infinite-invariance} rely on the assumption that the vector field $\bF$ defining the closed-loop dynamics \eqref{eq:closed-loop} is locally Lipschitz. To rectify this, the following result outlines weaker conditions for forward invariance of \eqref{eq:safe-set-infinite-constraints} for continuous closed-loop dynamics.

\begin{theorem}[Safety guarantee under weaker conditions]\label{thm:safety-under-weaker-conditions}
    Consider system~\eqref{eq:closed-loop} and the set $\Sc$ defined with an infinite number of inequality constraints as in~\eqref{eq:safe-set-infinite-constraints}, where $h$ is continuous in $\tau$ and continuously differentiable in $\bx$. Suppose $\cal T$ is compact, Assumption~\ref{assump:well-posed} holds, and $\bF$ is continuous. Then $\cal S$ is forward invariant for~\eqref{eq:closed-loop} if there exists an open neighborhood $\mathcal{N}\supset\partial\mathcal{S}$ such that for each $\bx\in\mathcal{N}\backslash\Sc$:
    \begin{equation}\label{eq:continuous-invariance-condition}
        \pdv{h}{\bx}(\tau,\bx) \bF(\bx) \geq 0,~\forall \tau\in\mathcal{M}(\bx)\coloneqq \argmin_{\tau\in\mathcal{T}} h(\tau,\bx).
    \end{equation}
    % where $\mathcal{M}(\bx)\coloneqq \argmin_{\tau\in\mathcal{T}} h(\tau,\bx)$.
\end{theorem}

\begin{proof}
    Since $\bF$ is continuous, it follows from Peano's Existence Theorem \cite[Thm. 1.1.2]{LakshmikanthamLeela} that for each initial condition $\bx_0$ the ODE $\dot{\bx}=\bF(\bx)$ admits at least one continuously differentiable solution. 
    % {\color{blue} Alternative proof:
    % For the sake of contradiction, suppose that $\Sc$ is not forward invariant.
    % That is, suppose that $\bx_0\in\Sc$ but there exists at least one solution that leaves $\Sc$. 
    % Since any solution $t\to\bx(t)$ is continuous, this implies that there must exist $\tau_*$ and
    % times $t_2 > t_1 \geq 0$ for which $h(\tau_*,\bx(t_1))=0$ and $h(\tau_*,\bx(t))>0$ for $t\in[t_1,t_2]$.
    % }
    %
    Define $H(\bx) \coloneqq \min_{\tau\in\mathcal{T}}h(\tau,\bx),$
    % \begin{equation}
    %     H(\bx) \coloneqq \min_{\tau\in\mathcal{T}}h(\tau,\bx),
    % \end{equation}
    where the minimum exists and $\bx\mapsto H(\bx)$ is continuous since $h$ is continuous and $\mathcal{T}$ is compact, 
    allowing $\mathcal{S}$ to be represented as
    $\mathcal{S} = \{\bx\in\R^n\mid H(\bx)\geq0\}$.
    Consider any solution $t\mapsto\bx(t)$ such that $\dot{\bx}(t)=\bF(\bx(t))$ on some interval of existence.
    Since $t\mapsto \bx(t)$ is continuous, the composition $\eta(t)\coloneqq H(\bx(t))$ is continuous. 
    % This definition allows $\mathcal{S}$ to be equivalently represented as
    % $\mathcal{S} = \{\bx\in\R^n\mid H(\bx)\geq0\}$.
    For the sake of contradiction, suppose that $\mathcal{S}$ is not forward invariant; that is, suppose that $\bx_0\in\mathcal{S}$ but there exists at least one solution of $\bF$ that leaves $\mathcal{S}$. Since any solution $t\mapsto \bx(t)$ is continuous, this implies there must exist a time $t_1\geq0$ at which $\bx(t_1)\in\partial \mathcal{S}$ and $\eta(t_1)=0$. Further, if such a solution were to leave $\mathcal{S}$, there must exist another time $t_2>t_1$ such that $\bx(t)\notin\mathcal{S}$ and $\eta(t)<0$ for all $t\in(t_1,t_2]$. Since $\bx(t_1)\in\partial\mathcal{S}$, $\mathcal{N}\supset\partial\mathcal{S}$ is open, and $t\mapsto \bx(t)$ is continuous, there also exists sufficiently small $\delta>0$ such that $\bx(t)\in\mathcal{N}\backslash\Sc$ for all $t\in(t_1,t_1 + \delta)$. Since $t\mapsto \eta(t)$ is continuous 
    % but not necessarily continuously differentiable, 
    let:
    \begin{equation}
        D_{+}\eta(t) \coloneqq \liminf_{s\rightarrow 0^+} (\eta(t + s) - \eta(t))/s,
    \end{equation}
    denote the lower right Dini derivative of $\eta$. 
    % which coincides with the standard derivative $\dot{\eta}(t)$ at any point where $\eta$ is differentiable. 
    It follows from Danskin's Theorem \cite[Prop. B.22]{Bertsekas} that $\eta(t)$ is differentiable whenever $\mathcal{M}(\bx(t))$ is a singleton; otherwise, it follows from Danskin's Theorem that:
    \begin{equation}
        D_{+}\eta(t) = \min_{\tau\in\mathcal{M}(\bx(t))}\pdv{h}{\bx}(\tau, \bx(t))\bF(\bx(t)).
    \end{equation}
    As $\bx(t)\in\mathcal{N}$ for $t\in(t_1,t_1 + \delta)$, it follows from \eqref{eq:continuous-invariance-condition} that:
    \begin{equation}
        D_{+}\eta(t) \geq 0,\quad \forall t\in(t_1,t_1 + \delta).
    \end{equation}
    Since $\eta(t)$ is continuous and the above holds, $\eta$ is monotonically non-decreasing on $(t_1,t_1 + \delta)$ \cite[Ch. 2.2.3]{Blanchini}. Since $\eta(t_1)=0$, this implies that $\eta(t)\geq0$ for all $t\in(t_1,t_1 + \delta)$, contradicting the initial assumption that $\eta(t)<0$ for all $t\in(t_1,t_2]$ for some $t_2>t_1$. Hence, by contradiction, $\mathcal{S}$ must be forward invariant, as desired. 
\end{proof}

Theorem~\ref{thm:safety-under-weaker-conditions} can be used to show that the controller obtained from the OD-CBF QP~\eqref{eq:infinite-OD-CBF-QP} guarantees safety.
\begin{corollary}[OD-CBF QP guarantees safety]
    Let the hypotheses of Proposition~\ref{prop:reduction-one-OD-variable} hold and consider the controller $\bk$ obtained from~\eqref{eq:infinite-OD-CBF-QP}. Then, $\bk$ renders $\Sc$ forward invariant.
\end{corollary}
\begin{proof}
    By Proposition~\ref{prop:reduction-one-OD-variable}, there exists a neighborhood $\Dc$ of $\Sc$ for which~\eqref{eq:infinite-OD-CBF-QP} is strictly feasible for all $\bx\in\Dc$.
    Now, for $\bx\in\Dc\setminus\Sc$, we have that $h(\tau,\bx) < 0$ for $\tau \in \mathcal{M}(\bx) = \argmin_{\tau\in\Tc}h(\tau,\bx)$.
    This implies that for $\bx\in\Dc\setminus\Sc$, 
    \begin{align*}
        \frac{\partial h}{\partial \bx}(\tau,\bx) (\bf(\bx)+\bg(\bx)\bk(\bx)) \geq -\theta(\bx)\alpha(h(\tau,\bx)) > 0,
    \end{align*}
    for all $\tau\in\mathcal{M}(\bx)$.
    Hence, the assumptions of Theorem~\ref{thm:safety-under-weaker-conditions} hold and $\bk$ renders $\Sc$ forward invariant.
\end{proof}

While the previous results demonstrate that a continuous controller meeting an infinite collection of constraints can be synthesized using optimization---and that, under additional assumptions, continuity is sufficient to establish invariance---practically synthesizing such controllers
% Constructive procedures, such as Sontag's formula, for synthesizing a controller meeting these constraints 
is complicated by the fact that there are infinitely many constraints.
In practice, one must reduce this infinite collection to a finite one. The following result outlines a procedure for obtaining a finite collection of constraints which are guaranteed to be feasible and whose satisfaction implies that of the infinite collection.

\begin{proposition}[Reduction to finite number of constraints]\label{prop:finite-constraints}
    Let $\Tc\subset\real^p$ and $\Xc\subset\real^n$ be compact and consider locally Lipschitz functions $a:\Tc\times\Xc\to\real$, $\bb:\Tc\times\Xc\to\real^{1\times m_{\bnu}}$. 
    Suppose that for each $\bx\in\Xc$, there exists $\bnu\in\real^{m_{\bnu}}$ satisfying~\eqref{eq:infinite-inequalities} strictly.
    Then, there exist positive constants $\epsilon_*$, $\Delta_*$ and $M_*$ such that by taking $\{ \tau_i \in \Tc \}_{i=0}^N$ with $N\in\mathbb{N}$, $\bigcup_{i=1}^N \Bc_{\Delta}(\tau_i)~\footnote{Given $\tau_i\in\real^p$ and $\Delta > 0$, $\Bc_{\Delta}(\tau_i)$ denotes the ball with radius $\Delta$ centered at $\tau$, i.e., $\Bc_{\Delta}(\tau_i) := \setdef{\tau\in\real^p}{\|\tau-\tau_i \| \leq \Delta}$.} \supseteq \Xc$ and $\Delta \leq \Delta_{*}$, we have that: 
    % the inequalities 
    \begin{align}\label{eq:robust-finite-inequalities}
        a(\tau_i,\bx) + \bb(\tau_i,\bx) \bnu \leq -\epsilon_*, \quad i\in[N],
    \end{align}
    are feasible for all $\bx\in\Xc$ and any controller $\bnu^*:\Xc\to\real^{m_{\bnu}}$ satisfying them for all $\bx\in\Xc$ and such that $\| \bnu^*(\bx) \| \leq M_*$ for all $\bx\in\Xc$ also satisfies~\eqref{eq:infinite-inequalities} for all $\bx\in\Xc$.
\end{proposition}
\begin{proof}
    Since~\eqref{eq:infinite-inequalities} are strictly feasible on $\Xc$, by Proposition~\ref{prop:artstein-infinite-constraints}
    there exists a smooth function $\bk:\Xc\to\real^{m_{\bnu}}$ such that: 
    \begin{align*}
        a(\tau,\bx) + \bb(\tau,\bx)\bk(\bx) < 0, \ \forall(\tau,\bx)\in\Tc\times\Xc.
    \end{align*}
    As $a, \bb, \bk$ are continuous, $\Tc$, $\Xc$ are compact, and $\bk$ satisfies~\eqref{eq:infinite-inequalities} strictly for $\bx\in\Xc$, there exists $\epsilon_* > 0$ such that:
    \begin{align*}
        -\epsilon_* = \sup\limits_{ \tau\in\Tc, \bx\in\Xc } a(\tau,\bx) + \bb(\tau,\bx)\bk(\bx),
    \end{align*}
    as well as $M_* > 0$ such that $\sup_{x\in\Xc} \| \bk(\bx) \| \leq M_*$.
    On the other hand, since $a$ and $\bb$ are locally Lipschitz and $\Tc\times\Xc$ is compact, $a$ and $\bb$ are Lipschitz in $\Tc\times\Xc$ with constants $L_a$, $L_b$ respectively.
    Now, let $\Delta_* := \frac{\epsilon_*}{L_a + L_b M_*}$.
    Consider $\{ \tau_i \}_{i=1}^N$
    with $\bigcup_{i=1}^N \Bc_{\Delta}(\tau_i) \supseteq \Tc$, with $\Delta \leq \Delta_*$.
    First, note that
    ~\eqref{eq:robust-finite-inequalities} are feasible for each $\bx\in\Xc$, because $\bnu = \bk(\bx)$ satisfies them.
    Now, let $\bnu^*:\Xc\to\real^{m_{\bnu}}$ satisfy~\eqref{eq:infinite-inequalities} and $\| \bnu^*(\bx) \| \leq M_*$.
    Consider $\tau\in\Tc$ and $\bx\in\Xc$.
    Since $\bigcup_{i=1}^N \Bc_{\Delta}(\tau_i) \supseteq \Tc$, $\exists i\in[N]$ such that $\tau\in\Bc_{\Delta}(\tau_i)$. Further, since $\sup_{\bx\in\Xc} \| \bnu^*(\bx) \| \leq M_*$, we have:
    \begin{align*}
        &a(\tau,\bx) + \bb(\tau,\bx)\bnu^*(\bx) \leq \\
        &a(\tau_i,\bx) + \bb(\tau_i,\bx)\bnu^*(\bx) + (L_a+L_b M_*)\Delta.
    \end{align*}
    Now, since $\Delta \leq \frac{\epsilon_*}{L_a + L_b M_*}$, we have that $a(\tau,\bx) + \bb(\tau,\bx)\bnu^*(\bx) \leq 0$. The argument is valid for any $\tau\in\Tc$ and $\bx\in\Xc$, from where the result follows.
\end{proof}

\section{Application to Backup CBFs}
We now connect our results to backup CBFs~\cite{gurriet2020scalable,GurrietICCPS18,YuxiaoCDC21,TamasACC23}, a common example where safe sets are defined by an infinite number of constraints \eqref{eq:infinite-inequalities}.
Here, one considers a candidate safe set $\mathcal{S}_{\rm{c}}$ that is not necessarily control invariant, along with a more conservative ``backup" set $\mathcal{S}_{\rm{b}}$ defined as:
\begin{equation}\label{eq:Sb}
    \begin{aligned}
        \mathcal{S}_{\rm{c}} = & \{\bx\in\R^n\mid h_{\rm{c}}(\bx) \geq 0\}, \\ 
        \mathcal{S}_{\rm{b}} = & \{\bx\in\R^n\mid h_{\rm{b}}(\bx) \geq 0\}, 
    \end{aligned}
\end{equation}
where $h_{\rm{b}}\,:\,\R^n\rightarrow\R$ and $h_{\rm{c}}\,:\,\R^n\rightarrow\R$ are continuously differentiable. 
While $\mathcal{S}_{\rm{c}}$ is not assumed to be control invariant, the
backup set $\mathcal{S}_{\rm{b}}$ is and comes with a ``backup" controller $\bk_{\rm{b}}\,:\,\R^n\rightarrow\real^m$ that enforces the forward invariance of $\mathcal{S}_{b}$ in that $\flow_{t}^{\rm{b}}(\bx)\in\mathcal{S}_{\rm{b}}$ for all $\bx\in\mathcal{S}_{\rm{b}}$ and all $t\geq0$, where:
\begin{equation}\label{eq:backup-flow}
    \begin{aligned}
        \tpdv{\flow_t^{\rm{b}}}{t}(\bx) = & \bf(\flow_t^{\rm{b}}(\bx)) + \bg(\flow_t^{\rm{b}}(\bx))\bk_{\rm{b}}(\flow_t^{\rm{b}}(\bx)), \\ 
        % \flow_{0}^{\rm{b}}(\bx) = & \bx,
    \end{aligned}
\end{equation}
and $\flow_{0}^{\rm{b}}(\bx) = \bx$,
is the flow of the dynamics under the backup controller.
The main idea behind backup CBFs is to use $\mathcal{S}_{\rm{b}}$ and $\bk_{\rm{b}}$ as an initial seed for safety, and then expand $\mathcal{S}_{\rm{b}}$ to a less conservative control invariant set by studying the backup flow $\flow_{t}^{\rm{b}}(\bx)$,
% In particular, this backup flow 
which
is used to define the \emph{implicit} safe set:
\begin{equation}
    \mathcal{S}_{I} =  \left\{\bx\mid h_{\rm{b}}(\flow_{T}^{\mathrm{b}}(\bx))\geq0\wedge h_{\rm{c}}(\flow_{\tau}^{\mathrm{b}}(\bx))\geq0,\tau\in[0,T]  \right\}, \label{eq:SI} 
\end{equation}
% \begin{equation}
%     h_{I}(\bx) =  \min\left\{h_{\rm{b}}(\flow_{T}^{\mathrm{b}}(\bx)),\min_{s\in[0,T]}h_{\rm{c}}(\flow_{s}^{\mathrm{b}}(\bx)) \right\} \label{eq:hI}
% \end{equation}
with $T>0$,
which is the set of all $\bx$ from which $\flow_t^{\rm{b}}$ returns to $\mathcal{S}_{\rm{b}}$ after $T$ seconds while remaining in $\mathcal{S}_{\rm{c}}$ up until time $T$. 
This set is defined by an infinitely many constraints and satisfies $\mathcal{S}_{I}\subset\mathcal{S}_{\rm{c}}$ so that forward invariance of $\mathcal{S}_{I}$ is implies satisfaction of the safety specification characterized by $\mathcal{S}_{\rm{c}}$. 
Controllers enforcing invariance of $\mathcal{S}_{I}$ are synthesized via:
\begin{equation}\label{eq:backup-cbf-qp}
    \begin{aligned}
        \min_{\bu\in\mathcal{U}}\quad & \tfrac{1}{2}\|\bu - \bk_{\rm{d}}(\bx)\|^2 \\ 
        \mathrm{s.t.} \quad & \dot{h}_{\rm{c}}(\flow_{\tau}^{\mathrm{b}}(\bx),\bu) \geq - \alpha_{\rm{c}}(h_{\rm{c}}(\flow_{\tau}^{\mathrm{b}}(\bx)),~\tau\in[0,T], \\
        \quad & \dot{h}_{\rm{b}}(\flow_{T}^{\mathrm{b}}(\bx),\bu) \geq - \alpha_{\rm{b}}(h_{\rm{b}}(\flow_{T}^{\mathrm{b}}(\bx)),
    \end{aligned}
\end{equation}
where $\dot{h}_c$ and $\dot{h}_b$ are defined using the sensitivity of the flow as in~\cite{gurriet2020scalable,GurrietICCPS18,YuxiaoCDC21,TamasACC23}, and $\alpha_{\rm{c}}$, $\alpha_{\rm{b}}$ are extended class $\mathcal{K}$ functions.

% While Nagumo's Theorem is widely cited to establish invariance results in the backup CBF literature, the discussion on regularity conditions that ensure the barrier conditions in \eqref{eq:backup-cbf-qp} are meaningful for invariance are comparatively sparse. 
To connect the results of the previous sections to the backup CBF approach, we first represent $\mathcal{S}_{I}$ in the form required by our results. Let $\mathcal{T}=[0,T]\cup\{T + 1\}$ and define:
\begin{equation}
    h(\tau,\bx) \coloneqq 
    \begin{cases}
        h_{\rm{c}}(\flow_{\tau}^{\rm b}(\bx)) & \tau\in[0,T], \\ 
        h_{\rm b}(\flow_{T}^{\rm b}(\bx))) & \tau = T + 1,
    \end{cases}
\end{equation}
% where $\flow_{\tau}^\mathrm{b}$ is the backup flow, $h$ is from \eqref{eq:S}, and $h_b$ is from \eqref{eq:Sb},
which satisfies:
\begin{equation*}
    \begin{aligned}
        \min_{\tau\in \mathcal{T}}h(\tau, \bx) = & \min\left\{h_{\rm{b}}(\flow_{T}^{\mathrm{b}}(\bx)),\min_{\tau\in[0,T]}h_{\rm{c}}(\flow_{\tau}^{\mathrm{b}}(\bx)) \right\},
    \end{aligned}
\end{equation*}
enabling $\mathcal{S}_{I}$ to be equivalently represented as:
\begin{equation}
    \mathcal{S}_{I} = \left\{\bx\in\R^n\mid h(\tau,\bx) \geq 0,~\forall \tau\in\mathcal{T} \right\},
\end{equation}
as required by \eqref{eq:safe-set-infinite-constraints}. Note that $\mathcal{T}$ is compact, $\bx \mapsto h(\tau,\bx)$ is continuously differentiable, and $\tau \mapsto h(\tau,\bx)$ is continuous\footnote{This follows from the fact that $[0,T]\cap\{T+1\}=\emptyset$ and any function is necessarily continuous at an isolated point; see, e.g., \cite[Ch. 4.3]{Abbott}.} for all $\tau\in\mathcal{T}$, as required by the results in preceding sections. 

Although $\mathcal{S}_I$ satisfies some of the preliminary hypothesis of our results, the practical instantiation of the safety filter \eqref{eq:backup-cbf-qp} is limited by the need to find $\alpha_{\rm{c}}\in\mathcal{K}^e$ satisfying an infinite collection of constraints\footnote{Various results establish the existence of these class $\mathcal{K}$ functions when $\mathcal{S}_{I}$ is compact \cite{TamasACC23} but do not provide procedures for obtaining them.}. To illustrate this, consider a double integrator \eqref{eq:control-affine} with $\bx=(x_1,x_2)\in\R^2$, $u\in\R$, $f(\bx)=(x_2,0)$, $\bg(\bx)=(0,1)$, with the objective of ensuring that $|x_1|\leq 1$ with an input $|u|\leq1$. Using backup CBFs, this can be accomplished by defining $h_{\rm{c}}(\bx) = 1 - x_1^2$, $h_{\rm{b}}(\bx)=\rho - \bx\T\bP\bx$, and $\bk_{\rm{b}}(\bx)=\tanh(-\bK\bx)$, where $\bP$ solves the algebraic Riccati equation and $\bK$ is an LQR gain, so that $\mathcal{S}_{\rm{b}}$ corresponds to a small sublevel set of the LQR value function $V(\bx)=\bx\T\bP\bx$ and the backup controller is a saturated LQR controller, which ensures invariance of $\mathcal{S}_{\rm{b}}$ for sufficiently small $\rho$. This backup controller defines the backup flow as in \eqref{eq:backup-flow} with $T=2$, which defines $\mathcal{S}_I$ as in \eqref{eq:SI}.

The corresponding safety filter \eqref{eq:backup-cbf-qp} is implemented via discretization of $\mathcal{T}$ according to Proposition \ref{prop:finite-constraints} with $N=42$, $\alpha_{\rm{c}}(s)=\alpha_{\rm{b}}(s)=s/2$, and $\bk_{\rm{d}}(\bx)=1$, the results of which are shown in Fig. \ref{fig:sims} (left). These simple choices of extended class $\mathcal{K}$ functions lead to infeasibility of \eqref{eq:backup-cbf-qp} at points in $\mathcal{S}_{I}$ and cause the corresponding closed-loop dynamics to be ill-defined. To resolve this, we leverage the optimal-decay approach from \eqref{eq:infinite-OD-CBF-QP} wherein the coefficient on the $\alpha$'s is automatically determined by the controller. Keeping in line with the results presented herein, we use a single OD variable for the $N=42$ constraints. The results of applying the resulting OD-CBF safety filter to the double integrator are shown in Fig. \ref{fig:sims} (right), where the controller is feasible, and the closed-loop dynamics are well-defined for all $\bx$ in $\mathcal{S}_{I}$.

\begin{figure}
    \centering
    \includegraphics{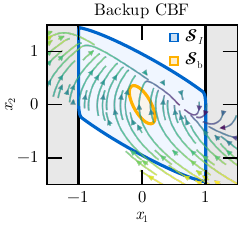}
    \hfill 
    \includegraphics[]{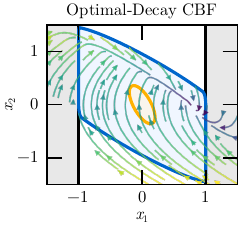}
    \vspace{-4mm}
    \caption{Closed-loop vector field the double integrator under the backup CBF controller \eqref{eq:backup-cbf-qp} [\textbf{left}] and backup CBF controller with an OD variable \eqref{eq:infinite-OD-CBF-QP} [\textbf{right}] with $\alpha_{\rm{c}}(s)=\alpha_{\rm{b}}(s)=s/2$ and $\theta_{\rm{d}}=1$. Regions not covered by the streamplot correspond to states where the closed-loop vector field is not defined due to infeasibility of the corresponding controller.}
    \vspace{-7mm}
    \label{fig:sims}
\end{figure}

\section{Conclusion}
Motivated by theoretical questions concerning backup CBFs, this paper presented a collection of results that extend CBF theory to safe sets defined by an infinite number of constraints. We outlined barrier-like conditions for invariance of these sets for autonomous systems and presented additional results for control systems that ensure the corresponding controllers enjoy desirable regularity properties.

\bibliographystyle{ieeetr}
\bibliography{biblio}

\end{document}

%% file: preamble.tex
\IEEEoverridecommandlockouts  
\usepackage{amsmath}
\usepackage{amssymb}
\usepackage{amsthm}
\usepackage{mathtools}
\usepackage{derivative}
\NewDerivative{\tpdv}{\partial}[style-frac=\tfrac]
\usepackage{xcolor}
\usepackage{bm}

\usepackage[noadjust]{cite}
\usepackage{url}

\usepackage{graphicx}
\usepackage[export]{adjustbox} % For aligning figures
\graphicspath{{figures/}}

\newtheorem{theorem}{Theorem}
\newtheorem{lemma}{Lemma}
\newtheorem{proposition}{Proposition}
\newtheorem{corollary}{Corollary}

\theoremstyle{definition}
\newtheorem{definition}{Definition}
\newtheorem{assumption}{Assumption}
\newtheorem{remark}{Remark}

\newcommand{\argmin}{\operatornamewithlimits{arg\,min}}

\newcommand{\R}{\mathbb{R}}

\newcommand{\Sc}{\mathcal{S}}

\newcommand{\T}{^\top}

\newcommand{\setdef}[2]{\{#1 \; | \; #2\}}

\newcommand{\bb}{\mathbf{b}}

\renewcommand{\bf}{\mathbf{f}} % NOTE: use \textbf if you really want to write bold non-math text.
\newcommand{\bg}{\mathbf{g}}

\newcommand{\bk}{\mathbf{k}}

\newcommand{\bu}{\mathbf{u}}
\newcommand{\bv}{\mathbf{v}}

\newcommand{\bx}{\mathbf{x}}

\newcommand{\bF}{\mathbf{F}}

\newcommand{\bK}{\mathbf{K}}

\newcommand{\bP}{\mathbf{P}}

\newcommand{\Bc}{\mathcal{B}}

\newcommand{\Dc}{\mathcal{D}}

\newcommand{\Kc}{\mathcal{K}}

\newcommand{\Nc}{\mathcal{N}}

\newcommand{\Tc}{\mathcal{T}}
\newcommand{\Uc}{\mathcal{U}}
\newcommand{\Vc}{\mathcal{V}}
\newcommand{\Xc}{\mathcal{X}}

\newcommand{\Wc}{\mathcal{W}}

\newcommand{\bphi}{\bm{\phi}}

\newcommand{\bnu}{\bm{\nu}}
\newcommand{\bvarphi}{\bm{\varphi}}

\newcommand{\des}{{\operatorname{d}}}

\usepackage{xcolor}
\definecolor{myblue}{RGB}{49, 114, 174}
\definecolor{myred}{rgb}{0.796, 0.235, 0.2}
\definecolor{mygreen}{rgb}{0.22, 0.596, 0.149}
\definecolor{mypurple}{rgb}{0.584,0.345,0.698}

\usepackage{blindtext}